\documentclass[11pt]{article}

\usepackage{amsmath,amssymb,amsbsy,amsfonts,amsthm,latexsym,
            amsopn,amstext,amsxtra,euscript,amscd,color}
\usepackage{verbatim}

\usepackage{caption}

\newfont{\teneufm}{eufm10}
\newfont{\seveneufm}{eufm7}
\newfont{\fiveeufm}{eufm5}

\newfam\eufmfam
              \textfont\eufmfam=\teneufm \scriptfont\eufmfam=\seveneufm
              \scriptscriptfont\eufmfam=\fiveeufm

\newtheorem{thm}{Theorem}
\newtheorem{lem}[thm]{Lemma}

\newtheorem{prop}[thm]{Proposition}
\newtheorem{rem}[thm]{Remark}
\newtheorem{defn}[thm]{Definition}

\newcommand{\Tr}{{\rm Tr}}
\newcommand{\Trn}{{\rm Tr}_n}
\newcommand{\Trm}{{\rm Tr}_m}
\newcommand{\cda}{{_cD_a}}

\def\+{\oplus}

\def\cW{{\mathcal W}}

\def\F{{\mathbb F}}

\def\F{{\mathbb F}}
\def\K{{\mathbb K}}

\def\00{{\bf 0}}
\def\11{{\bf 1}}
\def\+{\oplus}

\def\\{\cr}
\def\({\left(}
\def\){\right)}

\newcommand{\BBZ}{\mathbb{Z}}
\newcommand{\BBR}{\mathbb{R}}

\newcommand{\BBF}{\mathbb{F}}

\newcommand{\bwht}[2]{\mathcal{W}_{#1}(#2)}
\newcommand{\vwht}[3]{\mathcal{W}_{#1}(#2,#3)}

\newcommand{\cardinality}[1]{\# #1}

\providecommand{\newoperator}[3]{%
  \newcommand*{#1}{\mathop{#2}#3}}

\newoperator{\FD}{\mathrm{FD}}{\nolimits}

\begin{document}

\title{\bf  $C$-differentials, multiplicative uniformity and (almost) perfect $c$-nonlinearity}
\author{\Large P\aa l Ellingsen$^1$, Patrick Felke$^2$,  Constanza Riera$^1$, \\  \Large Pantelimon~St\u anic\u a$^3$, \Large Anton Tkachenko$^1$
\vspace{0.4cm} \\
$^1$Department of Computer Science,\\
 Electrical Engineering and Mathematical Sciences,\\
   Western Norway University of Applied Sciences,
  5020 Bergen, Norway;\\ {\tt \{pel, csr, atk\}@hvl.no}\\
$^2$University of Applied Sciences Emden-Leer, Constantiaplatz 4\\
 26723 Emden, Germany;  {\tt patrick.felke@hs-emden-leer.de}\\
$^3$Department of Applied Mathematics,
Naval Postgraduate School, \\
Monterey, CA 93943--5216,  USA; {\tt pstanica@nps.edu}
}
\date{\today}
\maketitle

\begin{abstract}
In this paper we define a new (output) multiplicative differential, and the corresponding $c$-differential uniformity.  With this new concept, even for characteristic $2$, there are perfect $c$-nonlinear (PcN) functions. We first characterize the $c$-differential uniformity of a function in terms of its Walsh transform. We further look at some of the known perfect nonlinear (PN)  and show that only one remains a PcN function, under a different condition on the parameters. In fact, the $p$-ary Gold PN function increases its $c$-differential uniformity significantly, under some conditions on the parameters. We then precisely characterize  the $c$-differential uniformity of the inverse function (in any dimension and characteristic), relevant for the Rijndael (and Advanced Encryption Standard) block cipher.
\end{abstract}
{\bf Keywords:} Boolean, $p$-ary functions, $c$-differentials, Walsh transform, differential uniformity, perfect and almost perfect $c$-nonlinearity\newline
{\bf MSC 2000}: 06E30, 11T06, 94A60, 94C10.

\section{Introduction and motivation}

In~\cite{BCJW02}, the authors used a new type of differential that is quite useful from a practical perspective for ciphers that utilize modular multiplication as a primitive operation. It is an extension of a type of differential cryptanalysis and it was used to cryptanalyse some existing ciphers (like a variant of the well-known IDEA cipher).
The authors argue that one should look (and some authors did) at other type of differentials for a Boolean (vectorial) function $F$ not only the usual $\left(F(x+a),F(x)\right)$. In~\cite{BCJW02}, the differential used in their attack was $\left(F(cx),F(x) \right)$.
Drawing inspiration from the mentioned  successful attempt, we therefore here start a theoretical analysis of an (output) multiplicative differential.
We first connect the differential uniformity (under this new concept) to the Walsh coefficients.
We next investigate  some of the known perfect nonlinear $p$-ary functions and show that with one exception (under a different condition on the parameters, though) they do not remain perfect nonlinear under the new concept. We also look at the Rijndael inverse function and its $c$-differential uniformity (for example, we show that in some instances the uniformity drops to $3$).

The objects of this study are Boolean and $p$-ary functions (where $p$ is an odd prime) and some of their differential properties.  We will introduce here only some needed notions, and the reader can consult~\cite{Bud14,CH1,CH2,CS17,MesnagerBook,Tok15} for more on Boolean and $p$-ary functions.

Let $n$ be a positive integer and $\F_{p^n}$ denote the  finite field with $p^n$ elements, and $\F_{p^n}^*=\F_{p^n}\setminus\{0\}$ (for $a\neq 0$, we often write $\frac{1}{a}$ to mean the inverse of $a$ in the considered finite field). Further, let $\F_p^m$ denote the $m$-dimensional vector space over $\F_p$.
We call a function from $\F_{p^n}$ to $\F_p$  a {\em $p$-ary Boolean function} on $n$ variables.
The cardinality of a set $S$ is denoted by $\cardinality{S}$.
For $f:\F_{p^n}\to \F_p$ we define the {\it Walsh-Hadamard transform} to be the integer-valued function
$\displaystyle
\bwht{f}{u}  = \sum_{x\in \F_{p^n}}\zeta_p^{f(x)-\Trn(u x)}, \ u \in \mathbb{F}_{p^n},
$
 where $\zeta_p= e^{\frac{2\pi i}{p}}$ and $\Trn:\F_{p^n}\to \F_p$ is the absolute trace function, given by $\Tr(x)=\sum_{i=0}^{n-1} x^{p^i}$.
Given a $p$-ary Boolean function $f$, the derivative of $f$ with respect to~$a \in \F_{p^n}$ is the Boolean function
$
 D_{a}f(x) =  f(x + a)- f(x), \mbox{ for  all }  x \in \F_{p^n}.
$

For positive integers $n$ and $m$, any map $F:\F_p^n\to\F_p^m$ is called a vectorial $p$-ary Boolean function, or $(n,m)$-function. When $m=n$, $F$ can be uniquely represented as a univariate polynomial over $\F_{p^n}$ (using the natural identification of the finite field with the vector space) of the form
$
F(x)=\sum_{i=0}^{p^n-1} a_i x^i,\ a_i\in\F_{p^n}.
$
The algebraic degree of $F$ is then the largest Hamming weight of the exponents $i$ with $a_i\neq 0$. For an $(n,m)$-function $F$, we define the Walsh transform $\vwht{F}{a}{b}$ to be the Walsh-Hadamard transform of its component function ${\rm Tr}_m(bF(x))$ at $a$, that is,
\[
  \vwht{F}{a}{b}=\sum_{x\in\F_{p^n}} \zeta_p^{\Trm(bF(x))-\Trn(ax)}, \text{ where $a\in \F_{p^n}, b\in \F_{p^m}$.}
\]

For an $(n,n)$-function $F$, and $a,b\in\F_{p^n}$, we let $\Delta_F(a,b)=\cardinality{\{x\in\F_{p^n} : F(x+a)-F(x)=b\}}$. We call the quantity
$\Delta_F=\max\{\Delta_F(a,b)\,:\, a,b\in \F_{p^n}, a\neq 0 \}$ the {\em differential uniformity} of $F$. If $\Delta_F\leq \delta$, then we say that $F$ is differentially $\delta$-uniform. If $\delta=1$, then $F$ is called a {\em perfect nonlinear} ({\em PN}) function, or {\em planar} functions. If $\delta=2$, then $F$ is called an {\em almost perfect nonlinear} ({\em APN}) function. It is well known that PN functions do not exist if $p=2$.

\section{$c$-Differentials}

It is natural to reflect about how the functions would respond, not only to the usual derivative, but to a more general derivative. Our proposal is inspired from a practical differential attack developed  in~\cite{BCJW02}.

\begin{defn}
Given a $p$-ary $(n,m)$-function   $F:\F_{p^n}\to \F_{p^m}$, and $c\in\F_{p^m}$, the (multiplicative) $c$-derivative of $F$ with respect to~$a \in \F_{p^n}$ is the  function
\[
 _cD_{a}F(x) =  F(x + a)- cF(x), \mbox{ for  all }  x \in \F_{p^n}.
\]
(Note that, if   $c=1$, then we obtain the usual derivative, and, if $c=0$ or $a=0$, then we obtain a shift of the function.)
\end{defn}

For an $(n,n)$-function $F$, and $a,b\in\F_{p^n}$, we let $_c\Delta_F(a,b)=\cardinality{\{x\in\F_{p^n} : F(x+a)-cF(x)=b\}}$. In the following, we call the quantity
$_c\Delta_F=\max\left\{_c\Delta_F(a,b)\,:\, a,b\in \F_{p^n}, \text{ and } a\neq 0 \text{ if $c=1$} \right\}$\begin{footnote}{Including $a=0$ for the case $c\neq 1$,  the equation $F(x)-cF(x)=b$ is of course,  $F(x)=b(1-c)^{-1}$, so we are looking here at how close $F$ is to a permutation polynomial, and similarly in the case $c=0$ for any $a$.}\end{footnote}
the {\em $c$-differential uniformity} of~$F$. If $_c\Delta_F=\delta$, then we say that $F$ is differentially $(c,\delta)$-uniform. If $\delta=1$, then $F$ is called a {\em perfect $c$-nonlinear} ({\em PcN}) function (certainly, for $c=1$, they only exist for odd characteristic $p$; however, one wonders whether they can exist for $p=2$ for $c\neq1$, and we shall argue later that that is actually true). If $\delta=2$, then $F$ is called an {\em almost perfect $c$-nonlinear} ({\em APcN}) function. It is easy to see that if $F$ is an $(n,n)$-function, that is, $F:\F_{p^n}\to\F_{p^n}$, then $F$ is PcN if and only if $_cD_a F$ is a permutation polynomial.

Furthermore, any nonconstant affine function is PcN for any $c\neq1$. Let $F(x)=Ax+B,0\neq A,B\in\F_{p^n}$, and let $c\neq 1$. Then $_cD_a F(x)= \alpha$ (for some $\alpha$) is equivalent to $A(x+a)+B-Acx-cB=\alpha$, that is, $A(1-c)x=\alpha+(c-1)B-Aa$, which has a unique solution. By a similar argument, the function $F(x)=Ax^{p^k}+B,0\neq A,B\in\F_{p^n}$is PcN for any $c\neq1$. We shall call these {\em trivial} PcN functions.

The solutions to the differential equation $_cD_{a}F(x) =b$ for various $c,b$ are not independent as the next proposition shows.
\begin{prop}
Let $F:\F_{p^n}\to \F_{p^m}$, $b_1,\,b_2\in \F_{p^m}$,  and $c_1\neq c_2\in \F_{p^m}^*$. If  $x_0$ is a  solution for $_{c_1}D_aF(x)=b_1$, then $x_0$ is also a solution for $_{c_2}D_aF(x)=b_2$ if and only if $\displaystyle F(x_0)=\frac{b_1-b_2}{c_2-c_1}$.
\end{prop}
\begin{proof}
Assume that
\[
_{c_1} D_a F(x_0)=b_1 \text{ and } _{c_2} D_a F(x_0)=b_2.
\]
Then
\begin{align*}
b_1&=F(x_0+a)-c_1 F(x_0)\\
b_2&=F(x_0+a)-c_2 F(x_0),
\end{align*}
which, by subtracting, renders the claim identity. The reciprocal is immediate.
\end{proof}

\section{Characterizing $c$-differential uniformity  via the Walsh transform}

In this section, we shall be using a method of Carlet~\cite{Car18} (which generalized the classical result of Chabaud and Vaudenay~\cite{CV95}) connecting the differential uniformity of an $(n,m)$-function to its Walsh coefficients. Since there are some subtle differences between the classical differential uniformity and our concept, we shall be proving (in any characteristic) the following result and some of its consequences, using the techniques of~\cite{Car18,CV95}. Generalizing the usual convolution of two functions $f,g$ in two variables over some cartesian product of fields $\BBF\times \K$, namely, $\displaystyle (f\otimes g)(a,b)=\sum_{x\in F,y\in K} f(x,y) g(a+x,b+y)$, we define a generalized convolution of the Walsh transforms (with a twist) of $F$ by
\begin{align*}
&(\cW_F\cW^c_F)^{\otimes (j+1)} (0,0)\\
&=\sum_{\substack{u_1,\ldots,u_j\in \F_{p^n}\\ v_1,\ldots,v_j\in \F_{p^m}}} \left(\overline{\cW_F}\left(\sum_{i=1}^j u_i,\sum_{i=1}^j v_i \right)\cW_F\left(\sum_{i=1}^j u_i,c\sum_{i=1}^j v_i \right)\right.\\
&\qquad\qquad\qquad\qquad\qquad\qquad \left. \cdot  \prod_{i=1}^j \cW_F(u_i,v_i)\overline{\cW_F}(u_i,cv_i)\right)
\end{align*}
(observe that it is a convolution since the sum  of the variables in each component is $0$).
We show the next theorem, which extends~\cite{Car18,CV95} to odd characteristics, as well as to the  $c$-differential context.
\begin{thm}
\label{thm:charact2}
Let $c\in\F_{p^m}$ and $n,m,\delta$ be fixed  positive integers. Let $F$  be an $(n,m)$-function, that is, $F:\F_{p^n}\to\F_{p^m}$. Let $\phi_\delta(x)=\sum_{j\geq 0} A_j x^j$ be a polynomial over $\BBR$ such that $\phi_\delta(x)=0$ for $x\in\BBZ,1\leq x\leq \delta$, and $\phi_\delta(x)>0$, for $x\in\BBZ, x> \delta$.  We then have
\begin{align*}
p^{2n} A_0+\sum_{j\geq 1} p^{-j(m+n)} A_j \left(\cW_F \cW_F^c \right)^{\otimes (j+1)} (0,0)\geq 0,
\end{align*}
with equality if and only if $F$ is $c$-differentially $\delta$-uniform.
\end{thm}
\begin{proof}
Let $a\in\F_{p^n}, \gamma\in\F_{p^m}$ be arbitrary elements.  Certainly, the set $\{x\in\F_{p^n}\,|\, _cD_aF(x)=\gamma\}$ is empty if $\gamma$ is not of the form $_cD_aF(b)$, for some $b\in\F_{p^n}$. Therefore, it is enough to consider only the cardinality $n_F(a,b,c)=|\{x\in\F_{p^n}\,|\, _cD_aF(x)= {_cD_a}F(b)\}|>0$. From the imposed condition on the polynomial $\phi_\delta$, given $F:\F_{p^n}\to\F_{p^m}$, for all $a,b\in\F_{p^n}$, then,
\[
\displaystyle \sum_{j\geq 0} A_j \left(n_F(a,b,c) \right)^j\geq 0,
\]
with equality if and only if $_c\Delta_F(a,b)\leq\delta$. Consequently, running with all $a,b\in\F_{p^n}$, any $(n,m)$-function $F$ satisfies
\[
\sum_{j\geq 0} A_j \sum_{a,b\in\F_{p^n}} \left(n_F(a,b,c) \right)^j\geq 0,
\]
with equality if and only of $_c\Delta_F\leq\delta$.

Observe that
\[
\displaystyle n_F(a,b,c)=p^{-m} \sum_{x\in F_{p^n}, v\in F_{p^m}} \zeta_p^{\Trm(v({_cD_a}F(x)-{_cD_a}F(b)))},
\]
since
\begin{align}
\label{eq:trace}
  \sum_{v\in\F_{p^m}} \zeta_p^{\Trm(v\alpha)}=
\begin{cases}  0 & \text{ if } \alpha\neq 0\\
p^m & \text{ if } \alpha = 0.
\end{cases}
\end{align}

For a fixed $j\geq 1$, we then have
\begin{align*}
&\sum_{a,b\in\F_{p^n}} \left(n_F(a,b,c) \right)^j \\
&= p^{-jm} \sum_{a,b\in\F_{p^n}} \sum_{\substack{x_1,\ldots,x_j\in\F_{p^n}\\ v_1,\ldots,v_j\in \F_{p^m}}} \zeta_p^{\sum_{i=1}^j \Trm(v_i (\cda F(x_i)-\cda F(b))}\\
&= p^{-jm} \sum_{a,b\in\F_{p^n}} \sum_{\substack{x_1,\ldots,x_j\in\F_{p^n}\\ v_1,\ldots,v_j\in \F_{p^m}}} \zeta_p^{\sum_{i=1}^j \Trm(v_i (F(x_i+a)-cF(x_i)- F(b+a)+cF(b))} .
\end{align*}
We need to insert some factors in the above expressions to make up the Walsh coefficients.
By identity~\eqref{eq:trace},  $\displaystyle \sum_{u_i\in\F_{p^n}} \zeta_p^{\Trn (u_i(x_i+a-y_i))}=p^n$, if $y_i=x_i+a$ and $0$, otherwise. Similarly,
$\displaystyle \sum_{u_0\in\F_{p^n}} \zeta_p^{\Trn (u_0(d-a-b))}=p^n$, if $d=a+b$ and $0$, otherwise.  Therefore,
\allowdisplaybreaks
\begin{align*}
&\sum_{a,b\in\F_{p^n}} \left(n_F(a,b,c) \right)^j =p^{-mj} \ p^{-(j+1)n} \sum_{a,b,d\in\F_{p^n}} \sum_{\substack{x_1,\ldots,x_j,y_1,\ldots,y_j\in\F_{p^n}\\ u_0,\ldots,u_j\in \F_{p^n},v_1,\ldots,v_j\in \F_{p^m}}} \\
&\zeta_p^{\sum_{i=1}^j \left[\Trm(v_i (F(y_i)-F(d)-cF(x_i)+cF(b)))+\Trn(u_i(x_i+a-y_i))\right]+\Trn(u_0(d-a-b))}\\
&=p^{-j(m+n)-n}  \sum_{\substack{u_0,u_1,\ldots,u_j\in\F_{p^n}\\ v_1,\ldots,v_j\in \F_{p^m}}} \overline{\cW_F}\left(u_0,\sum_{i=1}^j v_i\right) \cW_F\left(u_0,c\sum_{i=1}^j v_i\right) \\
&\qquad\qquad\qquad \cdot \prod_{i=1}^j \left( \cW_F(u_i,v_i)\overline{\cW_F}(u_i,cv_i)\right) \sum_{a\in\F_{p^n}} \zeta_p^{\Trn \left(a\left(\sum_{i=1}^j u_i-u_0\right) \right)}\\
&=p^{-j(m+n)-n} \ p^n \sum_{\substack{u_1,\ldots,u_j\in\F_{p^n}\\ v_1,\ldots,v_j\in \F_{p^m}}} \overline{\cW_F}\left(\sum_{i=1}^j u_i,\sum_{i=1}^j v_i\right) \cW_F\left(\sum_{i=1}^j u_i,c\sum_{i=1}^j v_i\right) \\
&\qquad\qquad\qquad\qquad\qquad\qquad\qquad\qquad \cdot \prod_{i=1}^j \left( \cW_F(u_i,v_i)\overline{\cW_F}(u_i,cv_i)\right)\\
&=p^{-j(m+n)} \left(\cW_F\cW^c_F\right)^{\otimes (j+1)} (0,0).
\end{align*}
For $j=0$, $\sum_{a,b\in\F_{p^n}} \left(n_F(a,b,c) \right)^j =p^{2n}$ and
the theorem follows.
\end{proof}
\begin{rem}
For a fixed $\delta$, an example of such a polynomial $\phi_\delta$ satisfying the conditions of the above theorem is simply $\phi_\delta(x)=(x-1)(x-2)\cdots (x-\delta)$, which certainly satisfies the conditions $\phi_\delta(x)=0$, for $1\leq x\leq \delta$, $x\in\mathbb{Z}$ and 
$\phi_\delta(x)>0$, for $ x> \delta$, $x\in\mathbb{Z}$.
\end{rem}

While we can take other values, we will only  consider below the particular  cases of perfect and almost perfect $c$-nonlinear functions, as the results are quite nice.

\subsection{The case $\delta=1$ -- perfect $c$-nonlinear (PcN)}

We can take the polynomial $\phi_1(x)=x-1$, which certainly satisfies the conditions of Theorem~\ref{thm:charact2}. Thus $A_0=-1,A_1=1$ and the relation of Theorem~\ref{thm:charact2}  simplifies to
\[
-p^{2n}+p^{-(m+n)} \sum_{\substack{u\in\F_{p^n}\\ v\in\F_{p^m}}}| \cW_F(u,v)|^2 | \cW_F(u,cv)|^2\geq 0.
\]
Thus, we obtain the next result.
\begin{prop}
Let  $m,n$ be fixed  positive integers and $c\in\F_{p^m}$, $c\neq 1$. Let $F$  be an $(n,m)$-function. Then
\[
 \sum_{\substack{u\in\F_{p^n}\\ v\in\F_{p^m}}} | \cW_F(u,v)|^2 | \cW_F(u,cv)|^2\geq p^{3n+m},
 \]
 with equality if and only if $F$ is a perfect $c$-nonlinear (PcN).
\end{prop}

\subsection{The case $\delta=2$ -- almost perfect $c$-nonlinear (APcN)}

We can take the polynomial $\phi_1(x)=(x-1)(x-2)=x^2-3x+2$, which certainly satisfies the conditions of Theorem~\ref{thm:charact2}. Thus $A_0=2,A_1=-3,A_2=1$ and the relation of Theorem~\ref{thm:charact2} simplifies to
\begin{align*}
&2\cdot p^{2n}-3\cdot p^{-(m+n)} \sum_{\substack{u\in\F_{p^n}\\ v\in\F_{p^m}}}  | \cW_F(u,v)|^2 | \cW_F(u,cv)|^2\\
&\quad+p^{-2(m+n)}
 \sum_{\substack{u_1,u_2\in\F_{p^n}\\ v_1,v_2\in\F_{p^m}}}  \overline{\cW_F}(u_1+u_2,v_1+v_2)\cW_F(u_1+u_2,c(v_1+v_2))\\
 &\qquad\qquad\qquad\qquad  \cdot
  \cW_F(u_1,v_1)\cW_F(u_2,v_2) \overline{\cW_F}(u_1,cv_1)\overline{\cW_F}(u_2,cv_2)\geq 0.
\end{align*}
We then have the following result.
\begin{prop}
Let  $m,n$ be fixed  positive integers and $c\in\F_{p^m}$, $c\neq 1$. Let $F$  be an $(n,m)$-function. Then
\begin{align*}
&\sum_{\substack{u_1,u_2\in\F_{p^n}\\ v_1,v_2\in\F_{p^m}}}   \overline{\cW_F}(u_1+u_2,v_1+v_2)\cW_F(u_1+u_2,c(v_1+v_2))\\
 &\qquad\qquad\qquad  \cdot
  \overline{\cW_F}(u_1,v_1)\overline{\cW_F}(u_2,v_2) \cW_F(u_1,cv_1)\cW_F(u_2,cv_2)\\
  & \geq
 3\cdot p^{m+n}  \sum_{\substack{u\in\F_{p^n}\\ v\in\F_{p^m}}}  | \cW_F(u,v)|^2 | \cW_F(u,cv)|^2 - 2\cdot p^{2(2n+m)},
 \end{align*}
 with equality if and only if $F$ is an  almost perfect $c$-nonlinear (APcN).
\end{prop}

In the spirit of Berger et al.~\cite{Berger06} and Carlet~\cite{Carlet19}, we can also express the $c$-differential uniformity of an $(n,m)$-function $F$ in terms of the Walsh transform of the $c$-derivative of~$F$. We will omit the proof as it is similar to the one of the classical case.
\begin{thm}
Let $1\neq c\in\F_{p^m}$, $n,m,\delta$ be fixed  positive integers, and let $F$  be an $(n,m)$-function. Let $\phi_\delta(x)=\sum_{j\geq 0} A_j x^j$ be a polynomial over $\BBR$ such that $\phi_\delta(x)=0$ for $x\in\BBZ,1\leq x\leq \delta$, and $\phi_\delta(x)>0$, for $x\in\BBZ, x> \delta$.  We then have
\begin{align*}
p^{n} A_0+\sum_{j\geq 1} p^{-jm} A_j  \sum_{v_1,\ldots,v_j\in\F_{p^m}}\cW_{_cD_aF}\left(0,\sum_{i=1}^j v_i\right) \prod_{i=1}^j \cW_{_cD_aF}(0,v_i)\geq 0,
\end{align*}
with equality if and only if $F$ is $c$-differentially $\delta$-uniform.
\end{thm}

\section{$c$-Differential uniformity for some known $PN$ classes}

If $c=1$, the following are some of the known classes of PN~\cite{CS97,DY06} (recall that they exist if and only if $p$ is odd).
\begin{thm}
\label{thm-PN}
The following functions $:\F_{p^n}\to\F_{p^n}$ are perfect nonlinear:
\begin{itemize}
\item[$(1)$] $F(x)=x^2$ on $\F_{p^n}$.
\item[$(2)$]  $F(x)=x^{p^k+1}$ on $\F_{p^n}$  is PN if and only if $\frac{n}{\gcd(k,n)}$ is odd.
\item[$(3)$] $F(x)=x^{10}\pm x^6 - x^2$ is PN over $\F_{3^n}$ if and only if $n=2$ or $n$ is odd. In general, for $u\in \F_{3^n}$, $F(x)=x^{10}-u x^6 - u^2 x^2$ is PN  over $\F_{3^n}$ if $n$ is odd.
\item[$(4)$] $F(x)=x^{(3^k+1)/2}$  is PN over $\F_{3^n}$  if and only if $\gcd(k,n)=1$ and $n$ is odd.
\end{itemize}
\end{thm}
It is not surprising that the above monomials are the only known planar (PN) functions since it was shown by R\'onyai and Sz\"onyi~\cite{RS89} that all planar functions are quadratic and by Zieve~\cite{Z15} that if the degree $n$ of $\F_{p^n}$ is large enough, namely, $p^n\geq (k-1)^4$, then the only PN monomials $x^k$  are in  the list above. It is conjectured that these are the only PN power functions.

Below we will investigate the same functions regarding their $PcN$ property for $c\neq 1$.
We first start with a lemma, which is possibly known (but we could not find an appropriate reference).
\begin{lem}
\label{lem:gcd}
Let $p,k,n$ be integers greater than or equal to $1$ (we take $k\leq n$, though the result can be shown in general). Then
\begin{align*}
&  \gcd(2^{k}+1,2^n-1)=\frac{2^{\gcd(2k,n)}-1}{2^{\gcd(k,n)}-1},  \text{ and if  $p>2$, then}, \\
& \gcd(p^{k}+1,p^n-1)=2,   \text{ if $\frac{n}{\gcd(n,k)}$  is odd},\\
& \gcd(p^{k}+1,p^n-1)=p^{\gcd(k,n)}+1,\text{ if $\frac{n}{\gcd(n,k)}$ is even}.\end{align*}
Consequently, if either $n$ is odd, or $n\equiv 2\pmod 4$ and $k$ is even,   then $\gcd(2^k+1,2^n-1)=1$ and $\gcd(p^k+1,p^n-1)=2$, if $p>2$.
\end{lem}
\begin{proof}
We shall use below the well-known identity
\[
\gcd(p^r-1,p^n-1)=p^{\gcd(r,n)}-1.
\]
Using the above identity, we first write $\gcd(p^{2k}-1,p^n-1)=p^{\gcd(2k,n)}-1$, and since $\gcd(p^k-1,p^k+1)=2$, unless $p=2$, in which case $\gcd(2^k-1,2^k+1)=1$, then we get
\begin{align*}
\gcd(2^{2k}-1,2^n-1)&= \gcd(2^{k}-1,2^n-1)\cdot \gcd(2^{k}+1,2^n-1)\\
&= \left(2^{\gcd(k,n)}-1 \right) \gcd(2^{k}+1,2^n-1),
\end{align*}
and the first claim is shown.

We now assume that $p>2$. If $k=1$, we observe that
\[
\gcd(p+1,p^n-1)=\gcd(p+1,p^n-1+p+1)=\gcd(p+1,p^{n-1}+1).
\]
If $n$ is even, then $n-1=2r+1$ (for some $r$) is odd and using the decomposition $p^{2r+1}+1=(p+1)(p^{2r}+p^{2r-1}+\cdots+p+1)$, we see that $\gcd(p+1,p^{n-1}+1)=p+1$.
If $n$ is odd, then $n-1=2r$ (for some $r$)  is even and consequently we continue the displayed reduction and arrive at $\gcd(p+1,p^n-1)=\gcd(p+1,p^{2r}+1)=\gcd(p+1,p^{2r-2}+1)=\cdots=\gcd(p+1,p^0+1)=2$.

We next let $k\geq 2$, next.
Let $d=\gcd(n,k)$ and $n=d m,k=d \ell$, with $\gcd(m,\ell)=1$.
First observe that
\begin{equation}
\begin{split}
\label{eq;descent}
\gcd\left(p^{d\ell}+1,p^{dm}-1 \right)
&=\gcd\left(p^{d\ell}+1,p^{dm}+p^{d\ell} \right) \\
& =\gcd\left(p^{d\ell}+1,p^{d(m-\ell)}+1\right) \\
&=\gcd\left(p^{d\ell}+1,p^{d(m-\ell)}-p^{d\ell} \right) \\
&=\begin{cases}
\gcd\left(p^{d\ell}+1,p^{d(m-2\ell)}-1\right) &\text{if } m\geq 2\ell\\
\gcd\left(p^{d\ell}+1,p^{d(2\ell-m)}-1\right)  &\text{if } m<2\ell.
\end{cases}
\end{split}
\end{equation}
In both cases, we see that $|m-2\ell |<m$. We continue the process and apply Fermat's descent method below. We consider two cases.

\noindent
{\em Case $1:$ $m$ is odd}. Since $m$ is odd, the process above arrives at
\[
\gcd\left(p^{d\ell}+1,p^{dm}-1 \right) =\gcd\left(p^{d\ell}+1,p^d-1 \right).
\]
Now, we switch sides and concentrate on the first expression. We obtain
\begin{align*}
\gcd\left(p^d-1,p^{d\ell}+1 \right)
&=\gcd\left(p^d-1,p^{d\ell}+1-p^{d(\ell-1)}(p^d-1)\right)\\
&= \gcd\left(p^d-1, p^{d(\ell-1)}+1\right)=\cdots\\
&= \gcd\left(p^d-1,p^d+1\right)=2.
\end{align*}

 \noindent
{\em Case $2:$ $m$ is even}. Then the reduction from~\eqref{eq;descent} arrives at (recall that now, $\ell$ is odd)
\allowdisplaybreaks
\begin{align*}
\gcd\left(p^{d\ell}+1,p^{dm}-1 \right)
&=\gcd\left(p^{d\ell}+1,p^{2d}-1 \right)\\
&=\gcd\left(p^{2d}-1,p^{d\ell}+p^{2d}\right)\\
&=\gcd\left(p^{2d}-1,p^{d(\ell-2)}+1\right)=\cdots\\
&=\gcd\left(p^{2d}-1,p^d+1\right)=p^d+1.
\end{align*}

Lastly, if  $n$ be odd and $k$ arbitrary, then $\gcd(2k,n)=\gcd(k,n)$, and so the above used identity becomes
$\left(2^{\gcd(k,n)}-1 \right) \gcd(2^{k}+1,2^n-1)=2^{\gcd(2k,n)}-1=2^{\gcd(k,n)}-1$, rendering $\gcd(2^{k}+1,2^n-1)=1$.
If $n\equiv 2\pmod 4$ and $k$ is even, say $n=4t+2$, and $k=2\ell$ for some $t,\ell$, then $\gcd(2k,n)=\gcd(4\ell,4t+2)=2\gcd(\ell, 2t+1)=\gcd(2\ell,4t+2)=\gcd(k,n)$, and the displayed identity implies $\gcd(2^{k}+1,2^n-1)=1$. A similar analysis works for $p>2$, as well.
 \end{proof}
 It is easy to see that if $F(x)=x^n$, then $_c\Delta_F(a,b)=_c\Delta_F(1,b/a^n)$, so we shall be using this often below.
\begin{thm}\label{c-diff}
Let $F:\F_{p^n}\to \F_{p^n}$ be the monomial $F(x)=x^d$,  and $c\neq  1$ be fixed. The following statements hold:
\begin{enumerate}
\item[$(i)$] If $d=2$, then $F$ is  APcN, for all $c\neq 1$.
\item[$(ii)$] If $d=p^k+1$, then  $F$ is not PcN, for all $c\neq 1$. Moreover, when $(1-c)^{p^k-1}=1$
and ${n}/{\gcd{(n,k)}}$ is even, the $c$-differential uniformity   $_c\Delta_F\geq p^g+1$, where $g=\gcd(n,k)$.
\item[$(iii)$]   Let $p=3$. If $\displaystyle d=\frac{3^k+1}{2}$, then  $F$ is   PcN, for $c=-1$ if and only if $\displaystyle \frac{n}{\gcd(n,k)}$ is odd.
\item[$(iv)$] If $p=3$ and $F(x)=x^{10}-u x^6-u^2 x^2$,  the $c$-differential uniformity of $F$ is $_c\Delta_F\geq 2$.
\end{enumerate}
\end{thm}
\begin{proof}
We first take $d=2$ and consider the equation $cD_a F(x)=b$. Thus,
\begin{align*}
b=F(x+a)-c F(x)=(1-c)x^2+2ax+a^2,
\end{align*}
and since we have the choice of $b$ (for example, $b=0$), the above equation has two solutions if and only if $c\neq 1$, and therefore it is $APcN$.

We now consider $d=p^k+1$. The equation  $cD_a F(x)=b$ becomes
\begin{align*}
b&=(x+a)^{p^k+1}-c x^{p^k+1}\\
&= x^{p^k+1}+a^{p^k+1} +a x^{p^k} +a^{p^k} x - cx^{p^k+1}\\
&=(1-c) x^{p^k+1} +a\, x^{p^k} +a^{p^k} x +a^{p^k+1}.
\end{align*}
Assuming $c\neq 1$, and choosing $b=-\frac{c\, a^{p^k+1}}{1-c}$, the equation transforms into
\[
 \left((1-c)x+a\right) x^{p^k}+ \frac{a^{p^k}}{1-c}  \left((1-c)x+a\right)=0,
\]
 which is equivalent to
 \begin{align*}
0&=  \left((1-c)x+a\right) \left(x^{p^k}+\frac{a^{p^k}}{1-c}\right)\\
&= \left((1-c)x+a\right) \left(x+\frac{a}{d}\right)^{p^k},
 \end{align*}
 where $r$ is one of the $p^k$-roots of $1-c\neq 0$.

 If $1-c=r^{p^k}\neq  r $ (which always happens if $\gcd(k,n)=1$) there are at least two roots of the equation  $cD_a F(x)=b$ and consequently, $F$ is not $PcN$.

 If $1-c=r^{p^k}=d$ (equivalently, for $c\neq 1$, $(1-c)^{p^k-1}=1$), then we write the equation $cD_a F(x)=b$, for some $a\neq 0$, as (we first multiply it by $d/a^{p^k+1}$)
 \[
 (r x/a)^{p^k+1}+(r x/a)^{p^k}+r x/a+r(1-b/a^{p^k+1})=0,
 \]
 and relabeling $y=r x/a$, $b_1=r(1-b/a^{p^k+1})$, we then get
 \begin{equation}
 \label{eq:b1}
 y^{p^k+1}+y^{p^k}+y+b_1=0.
 \end{equation}
 If $b_1=0$, then the equation becomes $ y^{p^k+1}+y^{p^k}+y= y( y^{p^k}+y^{p^k-1}+1)=0$ with the obvious solution $y=0$.

 We next consider the roots of second equation, that is, $y^{p^k}+y^{p^k-1}+1=0$, or equivalently  (with $z=1/y$), $z^{p^k}+z+1=0$. In~\cite{Li78} it is shown that, a trinomial $z^{p^k}-az-b$ in $\F_{p^n}$ has either zero, one, or $p^g$ roots, where $g=\gcd(n,k)$.

 We can perhaps approach this equation directly, but we
can be more precise and use the method of~\cite{CM04}, which fixed some errors of~\cite{Li78} and made its results more accurate. We will recall what was shown in~\cite{CM04}. Let $f(z)=z^{p^k}-az-b$ in $\F_{p^n}$, $g=\gcd(n,k)$, $m=n/\gcd(n,k)$ and ${\rm Tr}_g$ be the relative trace from $\F_{p^n}$ to $\F_{p^g}$. For $0\leq i\leq m-1$, we define $t_i=\sum_{j=i}^{m-2} p^{n(j+1)}$, $\alpha_0=a,\beta_0=b$. If $m>1$ (note that, if  $m=1$, then $k=n$, so $F(x)=x^2$, treated  in case {\em{(i)}}), then, for $1\leq r\le m-1$, we set
 \[
 \alpha_r=a^{1+p^k+\cdots+p^{kr}} \text{ and } \beta_r=\sum_{i=0}^r a^{s_i} b^{p^{ki}},
 \]
 where $s_i=\sum_{j=i}^{r-1} p^{k(j+1)}$, for $0\leq i\leq r-1$ and $s_r=0$. The trinomial $f$ has no roots in $\F_{p^n}$ if and only if $\alpha_{m-1}=1$ and $\beta_{m-1}\neq 0$. If  $\alpha_{m-1}\neq 1$, then it has a unique root, namely $x=\beta_{m-1}/(1-\alpha_{m-1})$, and, if $\alpha_{m-1}=1,\beta_{m-1}=0$, it has $p^g$ roots in $\F_{p^n}$ given by $x+\delta\tau$, where $\delta\in\F_{p^g}$, $\tau$ is fixed in $F_{p^n}$ with $\tau^{p^k-1}=a$ (that is, a $(p^k-1)$-root of $a$), and, for any $e\in\F^*_{p^n}$ with ${\rm Tr}_g(e)\neq 0$, then
 $\displaystyle x=\frac{1}{{\rm Tr}_g(e)} \sum_{i=0}^{m-1} \left( \sum_{j=0}^i e^{p^{kj}}\right) a^{t_i} b^{p^{ki}}$.
 We could easily simplify some of these parameters using the sum of the geometric sequence, namely
 \[
 s_i=\frac{p^{k(r+1)}-p^{k(i+1)}}{p^k-1},\quad \alpha_r=a^{\frac{p^{k(r+1)}-1}{p^k-1}},
 \]
 though, for our case, these closed forms will not be useful.

 For our case, $a=b=-1$ and we further compute the involved parameters, splitting the analysis in two cases. We denote by $\sigma_t\equiv t\pmod 2$ the parity of $t$, that is $t\pmod 2\in\{0,1\}$. First, recall that $\alpha_0=\beta_0=-1$.

 \noindent
 {\em Case $1$.} $m=\frac{n}{\gcd{(n,k)}}=2\ell+1$, for some $\ell\in\mathbb Z$. Let $r=m-1$. Recall that $s_{m-1}=0$, thus  $\sigma_{s_{m-1}}=0$. The parities of $s_i$, $0\leq i\leq m-2$, are
$\sigma_{s_i}= (m-2-i+1)\pmod 2 = \sigma_i$ (since $m-1$ is even).
 Further,
 \begin{align*}
 \alpha_{m-1}&=(-1)^{1+p^k+\cdots+p^{k(m-1)}}=(-1)^{\sigma_m}=-1,\\
 \beta_{m-1}&=\sum_{i=0}^{m-1} (-1)^{s_i+p^{ki}} = \sum_{i=0}^{m-2} (-1)^{s_i+p^{ki}} + (-1)^{s_{m-1}+p^{k(m-1)}}\\
 &=-1+\sum_{i=0}^{m-2} (-1)^{\sigma_{i+1}}=-1.
 \end{align*}
 Therefore, we conclude that the equation $z^{p^k}+z+1=0$ has a unique solution in $\F_{p^n}$. Thus, when $\frac{n}{\gcd{(n,k)}}$, then $F(x)=x^{p^k+1}$ is not $PcN$, for all $c\neq 1$ with $(1-c)^{p^k-1}=1$, and consequently, for all $c\neq 1$, given our previous argument.

  \noindent
 {\em Case $2$.} $m=\frac{n}{\gcd{(n,k)}}=2\ell$, for some $\ell\in\mathbb Z$. Let $r=m-1$. Recall that $s_{m-1}=0$, thus  $\sigma_{s_{m-1}}=0$. The parities of $s_i$, $0\leq i\leq m-2$, are
$\sigma_{s_i}= (m-2-i+1)\pmod 2 = \sigma_{i+1}$ (since $m-1$ is odd).
 Further,
 \begin{align*}
 \alpha_{m-1}&=(-1)^{1+p^k+\cdots+p^{k(m-1)}}=(-1)^{\sigma_m}=1,\\
 \beta_{m-1}&=\sum_{i=0}^{m-1} (-1)^{s_i+p^{ki}} = \sum_{i=0}^{m-2} (-1)^{s_i+p^{ki}} + (-1)^{s_{m-1}+p^{k(m-1)}}\\
 &=-1+\sum_{i=0}^{m-2} (-1)^{\sigma_{i}}=0.
 \end{align*}
 Therefore, we infer that the equation $z^{p^k}+z+1=0$ has $p^g$ solutions in $\F_{p^n}$. Thus, the initial equation~\eqref{eq:b1} has $p^g+1$ solutions and so, the $c$-differential uniformity in this case (under $(1-c)^{p^k-1}=1$ and $\frac{n}{\gcd{(n,k)}}$ even) is at least $p^g+1$, where $g=\gcd(n,k)$.

 Let us treat now the case of $d=(3^k+1)/2$ under $c=-1$, in $\F_{3^n}$. As used in~\cite{CM04}, our function is $PcN$ if and only if the $c$-derivative $(x+1)^{\frac{3^k+1}{2}}-c x^{\frac{3^k+1}{2}}$ is a permutation polynomial   if and only if $h_c(x)=(x-1)^{\frac{3^k+1}{2}}-c (x+1)^{\frac{3^k+1}{2}}$ is a permutation polynomial. Since $2|3^n-1$, for all $n$, then we can always write $x=y+y^{-1}$, for some $y\in\F_{3^n}$. Our condition (for general $c\neq 1$) becomes
 \begin{align*}
 h_c(x)&=(y+y^{-1}-1)^{\frac{3^k+1}{2}}-c(y+y^{-1}+1)^{\frac{3^k+1}{2}}\\
 &= \frac{(y^2-y+1)^{\frac{3^k+1}{2}}-c(y^2+y+1)^{\frac{3^k+1}{2}}}{y^{\frac{3^k+1}{2}}}\\
 &= \frac{(y+1)^{3^k+1}-c(y-1)^{3^k+1}}{y^{\frac{3^k+1}{2}}}\\
 &= \frac{(1-c)y^{3^k+1}+(1+c)y^{3^k}+(1+c) y+(1-c)}{y^{\frac{3^k+1}{2}}}\\
 &=(1-c) y^{\frac{3^k+1}{2}}+(1+c) y^{\frac{3^k-1}{2}}+(1+c) y^{\frac{-3^k+1}{2}}+(1-c) y^{\frac{-3^k-1}{2}}\\
 &= (1-c) T_{\frac{3^k+1}{2}}(y)+(1+c) T_{\frac{3^k-1}{2}}(y)
 \end{align*}
 is a permutation polynomial ($T_\ell$ is the Chebyshev polynomial of the first kind). If $c=-1$, we obtain that
 $T_{\frac{3^k+1}{2}}(y)$ must be a permutation polynomial which happens~\cite{No68} if and only if $\displaystyle \gcd\left( \frac{3^k+1}{2}, 3^{2n}-1 \right)=1$ if and only if $\displaystyle \gcd\left(  3^k+1 , 3^{2n}-1 \right)=2$. By Lemma~\ref{lem:gcd} a necessary and sufficient condition for that to happen is for $\displaystyle \frac{n}{\gcd(n,k)}$ to be odd.

 We now consider the function
 $F(x)=x^{10}-u x^6-u^2 x^2$ over $\mathbb{F}_{3^n}$. The equation $_cD_a F(x)=b$ becomes
  \begin{align*}
 &  (1-c) x^{10}+a x^9+  u(c-1) x^6 +u a^3 x^3+  u^2(c-1) x^2\\
  &\qquad \qquad \qquad +(a^9+a u^2) x+ a^{10}-u a^6 - u^2 a^2-b=0.
 \end{align*}
 Taking $a=0, b=(c-1)(u^2+u-1)$ the above equation will have solutions $x=1,2$ and so, $_c\Delta_F\geq 2$ (surely, one can take even nonzero values of $a$, in many instances, if not all; for example, for $c=2$ and $u=1$, we can take $a=b=2$, rendering the solutions $x=0,1$).
\end{proof}

\subsection{Some computational data}

Table 1 shows the maximal $c$-differential uniformity (for $a,c\neq0$) for the Gold function $F(x)=x^{2^k+1}$ and the Kasami function $G(x)=x^{2^{2k}-2^k+1}$ over $\mathbb{F}_{2^n}$, for $k=2$.
Note that, for $k=1$, these two functions are equal. 
The maximal $c$-differential uniformity (for $a,c\neq0$), taking $k=1$ for both functions, is equal to 2 for $n\geq 2$, and equal to 3 for $n\geq3$, which can be argued theoretically. The cases $n=1,2$ are straightforward. Let $n\geq3$.
The $c$-derivative of $F(x)=x^3$ over $\mathbb{F}_{2^n}$ is
 \[
_cD_1F(x)=(1+c)x^3+x^2+x+1.
\]
 Taking $b=1$, the equation $_cD_1F(x)=b$ is equivalent to $x((1+c)x^2+x+1)=0$. This equation always has $x=0$ as a solution, while the quadratic equation has two solutions if and only if $Tr(1+c)=0$. Taking $c=\alpha^2+\alpha+1\neq0$ for $n\geq3$, where $\alpha$ is a primitive root of    $\F_{2^n}$, we obtain three solutions to the equation $_cD_aF(x)=b$, and therefore the maximal $c$-differential uniformity (for $c\neq0$) for the Gold/Kasami function $F(x)=x^3$ over $\mathbb{F}_{2^n}$ is~3.

For $k=2$ and $n\geq3$, the results for the Gold and Kasami functions (respectively, $F(x)=x^5$ and $G(x)=x^{13}$) show a maximal $c$-differential uniformity (over $a,c\neq0$) of 3 for $n$ odd (i.e. for $\gcd(n,k)=1$), and of 5 for $n$ even. Note that, by Theorem~\ref{c-diff}$(ii)$, taking $p=2$ and $n\equiv 0\pmod 4$, the result is shown for the Gold function with $k=2$, since, for these cases, the lower bound is 5, and, since the degree of the function is also 5, the $c$-differential uniformity must be exactly 5  for all $c$ such that $(1-c)^3=1$. Hence, the maximal $c$-differential uniformity (for $c\neq0$) for the Gold function is $5$, when $k=2$. We simply double checked computationally (for small dimensions) our proof of Theorem~\ref{c-diff}$(ii)$. It would be interesting to prove theoretically whether our observations hold also for other dimensions, and for~$x^{13}$, or perhaps, for the general Kasami functions.

\begin{table}[]
\centering
\begin{tabular}{ccl}
\hline
\multicolumn{1}{|c|}{n} & \multicolumn{1}{c|}{Gold function, $k=2$} & \multicolumn{1}{l|}{Kasami function, $k=2$} \\ \hline
\multicolumn{1}{|c|}{1} & \multicolumn{1}{c|}{2} & \multicolumn{1}{c|}{2} \\ \hline
\multicolumn{1}{|c|}{2} & \multicolumn{1}{c|}{4} & \multicolumn{1}{c|}{4} \\ \hline
\multicolumn{1}{|c|}{3} & \multicolumn{1}{c|}{3} & \multicolumn{1}{c|}{3} \\ \hline
\multicolumn{1}{|c|}{4} & \multicolumn{1}{c|}{5} & \multicolumn{1}{c|}{5} \\ \hline
\multicolumn{1}{|c|}{5} & \multicolumn{1}{c|}{3} & \multicolumn{1}{c|}{3} \\ \hline
\multicolumn{1}{|c|}{6} & \multicolumn{1}{c|}{5} & \multicolumn{1}{c|}{5} \\ \hline
\multicolumn{1}{|c|}{7} & \multicolumn{1}{c|}{3} & \multicolumn{1}{c|}{3} \\ \hline
\multicolumn{1}{|c|}{8} & \multicolumn{1}{c|}{5} & \multicolumn{1}{c|}{5} \\ \hline
\end{tabular}
\caption{Gold and Kasami functions, $k=2$.}
\label{table1}
\end{table}

Table~\ref{table2} shows the maximal $c-$differential uniformity for the functions $F(x)=x^{10}\pm x^6-x^2$ over $\mathbb{F}_{3^n}, c\in\mathbb{F}_{3^n}\setminus\{0,1\}$. The tests indicate the following behaviour for both polynomials: The maximal $c$-differential uniformity (for $c\neq0,1$) is 2 if $n=2$, $n+1$ for $n=1,3,5$ and 10 for $n\geq 7, n$ odd.
It would be nice to have a proof for the general case $x^{10}-ux^6-u^2x^2, u\in\mathbb{F}_{3^n}$ especially showing that the maximal $c$-differential uniformity is 10 for $n\geq 7,u=\pm1 $ and to see the mathematical reason for jump from 6 to 10 in Table~\ref{table2}.

\begin{table}[h]
\centering
\begin{tabular}{ccc}
\hline
\multicolumn{1}{|c|}{$n$} & \multicolumn{1}{c|}{$F(x)=x^{10}-x^6-x^2$}& \multicolumn{1}{c|}{$F(x)=x^{10}+x^6-x^2$}\\ \hline
\multicolumn{1}{|c|}{1} & \multicolumn{1}{c|}{2} &\multicolumn{1}{c|}{2} \\ \hline
\multicolumn{1}{|c|}{2} & \multicolumn{1}{c|}{2}&\multicolumn{1}{c|}{2} \\ \hline
\multicolumn{1}{|c|}{3} & \multicolumn{1}{c|}{4}&\multicolumn{1}{c|}{4} \\ \hline
\multicolumn{1}{|c|}{5} & \multicolumn{1}{c|}{6}&\multicolumn{1}{c|}{6} \\ \hline
\multicolumn{1}{|c|}{7} & \multicolumn{1}{c|}{10}&\multicolumn{1}{c|}{10} \\ \hline
\multicolumn{1}{|c|}{9} & \multicolumn{1}{c|}{10}&\multicolumn{1}{c|}{10} \\ \hline
\multicolumn{1}{|c|}{11} & \multicolumn{1}{c|}{10}&\multicolumn{1}{c|}{10} \\ \hline
\end{tabular}
\caption{$F(x)=x^{10}\pm x^6-x^2$ over $\mathbb{F}_{3^n}, c\in\mathbb{F}_{3^n}\setminus\{0,1\}$.}
\label{table2}
\end{table}

%

\section{$c$-Differential uniformity for the inverse function}

 Since there has been quite a bit of effort to investigate the inverse function over $\F_{2^n}$ as it is relevant in Rijndael and Advance Encryption Standard, it is natural to wonder how it behaves with respect to $c$-differential uniformity.
 We will need the following lemma (see~\cite{BRS67}, for the first part; the second part is probably known and easy to derive by completing the square in the given quadratic equation, that is, writing $(2ax+b)^2=a^2-4b$).
\begin{lem}
\label{lem10}
Let $n$ be a positive integer. We have:
\begin{enumerate}
 \item[$(i)$] The equation
$x^2 + ax + b = 0$, with $a,b\in\F_{2^n}$, $a\neq 0$,
has two solutions in $\F_{2^n}$ if  $\Tr\left(
\frac{b}{a^2}\right)=0$, and zero solutions otherwise.
\item[$(ii)$]  The equation
$x^2 + ax + b = 0$, with $a,b\in\F_{p^n}$, $p$ odd,
has (two, respectively, one) solutions in $\F_{p^n}$ if and only if the discriminant $a^2-4b$ is a (nonzero, respectively, zero) square in $\F_{p^n}$.
\end{enumerate}
\end{lem}

 \subsection{The inverse function in even characteristic}
 We first treat the even  characteristic.
 \begin{thm}
 Let $n$ be a positive integer, $1\neq c\in\F_{2^n}$ and $F:\F_{2^n}\to\F_{2^n}$ be the inverse function defined by $F(x)=x^{2^n-2}$. We have:
 \begin{enumerate}
 \item[$(i)$] If $c=0$, then $F$ is PcN (that is, $F$ is a permutation polynomial).   \item[$(ii)$] If $c\neq 0$ and $\Trn(c)=\Trn(1/c)=1$, the $c$-differential uniformity of $F$ is $2$ $($and hence $F$ is APcN$)$.
 \item[$(iii)$] If $c\neq 0$ and $\Trn(1/c)=0$, or $\Trn(c)=0$, the $c$-differential uniformity of $F$ is $3$.
  \end{enumerate}
 \end{thm}
 \begin{proof}
 We start with the $c$-differential uniformity equation at $a$, namely, $(x+a)^{2^n-2}+cx^{2^n-2}=b$.

 If $c=0$, we have at most one solution for the $c$-differential equation, since then the above  $c$-differential uniformity equation  becomes $(x+a)^{2^n-2}=b$. If $b=0$, then $x=a$ is the only solution, otherwise, multiplying by $x+a$, we get $b(x+a)=1$, which has also only one solution. The proof of $(i)$ is done.

 So, next,  we may assume $c\neq 0$.
 If $a=0$, the equation becomes $(1+c) x^{2^n-2}=b$. If $b=0$, then $x=0$ is the only solution. If $b\neq 0$ (so $x\neq 0$), multiplying by $x$, the equation becomes $1+c=bx$, which has only one solution.

We next assume that $c\neq 0, a\neq 0$. Recall that,   if $F(x)=x^d$ and $a\neq0$, then $_c\Delta_F(a,b)=_c\Delta_F(1,b/a^d)$, so we can look at the equation 

\begin{equation}
\label{eq:inverse1}
 (x+1)^{2^n-2}+cx^{2^n-2}=b.
\end{equation}
 If $b=0$, we easily get only one solution, since the above equation is equivalent to $c(x+1)=x$, so we may assume below that $b\neq 0$.

\noindent
{\em Case $1$.} Let $b=1\neq c$. Then $x=0$ is a solution of~\eqref{eq:inverse1}. Assume now that $x\neq 0,1$ (certainly, $x= 1$ is not a solution). Multiplying~\eqref{eq:inverse1} by $x(x+1)$ we get $x+c(x+1)=x(x+1)$, which is equivalent to $x^2+cx+c=0$. By Lemma~\ref{lem10}$(i)$, this equation has two solutions (we would have only one solution if $c=0$)  if and only if $\Trn(1/c)=0$. Thus, altogether, we have three solutions for~\eqref{eq:inverse1} under  $\Trn(1/c)=0$. If  $\Trn(1/c)=1$, then ~\eqref{eq:inverse1}  has only the solution $x=0$.

\noindent
{\em Case $2$.} Let $b=c\neq 0,1$. Certainly, $x=1$ is a solution of~\eqref{eq:inverse1}, while $x=0$ is not a solution, so we assume now that $x\neq 0,1$. Again, multiplying~\eqref{eq:inverse1}  by $x(x+1)$ we get
$x+c(x+1)=cx(x+1)$, which is equivalent to $x^2+c^{-1} x+1=0$, which has two solutions if and only if $\Trn \left(1/c^{-2}\right)=\Trn(c^2)=\Tr(c)=0$. Thus, altogether, we have three solutions for~\eqref{eq:inverse1} under  $\Trn(c)=0$.

\noindent
{\em Case $3$.} Let $b\neq 1,c$ (so, $x\neq 0,1$). Multiplying~\eqref{eq:inverse1}  by $x(x+1)$ we get $x+c(x+1)=bx(x+1)$, which has a unique solution if $b=0,\,c\neq1$; otherwise, equation~\eqref{eq:inverse1} is equivalent to $x^2+\left(\frac{b+c+1}{b}\right) x+\frac{c}{b}=0$. If $b=c+1$, then we have a unique solution, otherwise, we have two solutions if and only if
$\Trn\left(\frac{bc}{b^2+c^2+1} \right)=0$. As we saw, it will be enough to consider the case of $\Trn(c)=\Trn(1/c)=1$, since otherwise, we have three solutions for some $b$. Below, we argue that we always can find some $b\neq 0$ for which this last trace, $\Trn\left(\frac{bc}{b^2+c^2+1} \right)=0$.

If $n$ is odd, then we claim that there exists some value of $u$ such that
\[\frac{uc}{u^2+c^2+1}=\frac{c+1}{c}=1+\frac{1}{c},\]
which follows from the fact that this last equation is equivalent to $(c+1) u^2+c^2 u+(c+1)^2=0$, that is, $u^2+\frac{c^2}{c+1} u+ (c+1)^2=0$, which, by Lemma~\ref{lem10}$(i)$, has solutions if and only if $0=\Trn\left(\frac{(c+1)^2}{c^4/(c+1)^2} \right)=\Trn\left(\frac{(c+1)^4}{c^4} \right)=\Trn\left(\frac{c+1}{c} \right)=\Trn\left(1+\frac{1}{c} \right)$. Since $\Trn(1/c)=1$ and $n$ is odd (thus, $\Trn(1)=1$), then $\Trn\left(1+\frac{1}{c} \right)=0$. Taking $b=u$ such solution, then the condition $\Trn\left(\frac{bc}{b^2+c^2+1} \right)=0$ will hold and consequently, we have two roots of the $c$-differential equation in this case.

We now let $n$ be even and consider the equation
$\displaystyle \frac{uc}{u^2+c^2+1}=1+c+\frac{1}{c},$
which is equivalent to  $\displaystyle u^2+\frac{c^2}{c^2+c+1} u+(c+1)^2=0$. This last equation has solutions (by Lemma~\ref{lem10}$(i)$) if and only if
$0=\Trn\left(\frac{(c+1)^2}{c^4/(c^2+c+1)^2} \right)=\Trn\left(\frac{(c+1)}{c^2/(c^2+c+1)} \right)=\Trn\left(\frac{(c+1)(c^2+c+1)}{c^2}\right)=\Trn\left(\frac{c^3+1}{c^2} \right)=\Trn(c)+\Trn(1/c^2)=\Trn(c)+\Trn(1/c)=0$, which will certainly happen. By taking $b$ to be equal to such a solution $u$, then, $\Trn\left(\frac{bc}{b^2+c^2+1} \right)=\Trn(1+c+1/c)=0$, and consequently, the $c$-differential equation will have two solutions.
 \end{proof}

\subsection{The inverse function in odd characteristic}
 We now treat the case of the inverse for odd characteristic.
 \begin{thm}
 Let $p$ be an odd prime, $n\geq 1$ be a positive integer, $1\neq c\in\F_{p^n}$ and $F:\F_{p^n}\to\F_{p^n}$ be the inverse $p$-ary function defined by $F(x)=x^{p^n-2}$. We have:
 \begin{enumerate}
 \item[$(i)$] If $c=0$, then $F$ is PcN (that is, $F$ is a permutation polynomial).  \item[$(ii)$] If $c\neq 0,4,4^{-1}$, $(c^2-4c)\in \left(\F_{p^n}\right)^2$, or $(1-4c)\in \left(\F_{p^n}\right)^2$, the $c$-differential uniformity of $F$ is $3$.
 \item[$(iii)$] If $c=4,4^{-1}$, the $c$-differential uniformity of $F$ is $2$ $($and hence $F$ is APcN$)$.
  \item[$(iv)$] If  $c\neq 0$, $(c^2-4c)\notin \left(\F_{p^n}\right)^2$ and $(1-4c)\notin \left(\F_{p^n}\right)^2$, the $c$-differential uniformity of $F$ is $2$ $($and hence $F$ is APcN$)$.
    \end{enumerate}

 \end{thm}
 \begin{proof}
 The proof for $(i)$ ($c=0$, as well as $a=0$) is similar to the one for characteristic 2.

We next assume that $c\neq 0$. For $a=0$, as in the case of even characteristic, the corresponding equation $(x+a)^{p^n-2}-cx^{p^n-2}=b$ has one solution. We assume now that $a\neq 0$. As in characteristic $2$, we only need to investigate the equation
\begin{equation}
\label{eq:inverse2}
 (x+1)^{p^n-2}-cx^{p^n-2}=b.
\end{equation}
 If $b=0$, we easily get only one solution, so we may assume below that $b\neq 0$.

\noindent
{\em Case $1$.} Let $b=1\neq c$. Then $x=0$ is a solution of~\eqref{eq:inverse2}. Assume now that $x\neq 0$ (also, $x\neq -1$, since then, $1=b=c$, which is impossible). Multiplying~\eqref{eq:inverse2} by $x(x+1)$ we get $x-c(x+1)=x(x+1)$, which is equivalent to $x^2+cx+c=0$. By Lemma~\ref{lem10}$(ii)$, this equation has solutions  if and only if the discriminant $D_1=c^2-4c\in \left(\F_{p^n}\right)^2$ (that is, $D_1$ is a square in $\F_{p^n}$); we have two solutions if $D_1\neq 0$ and one solution if $D_1=0$. Thus, altogether, we have three solutions for~\eqref{eq:inverse2} if  $0\neq D_1\in \left(\F_{p^n}\right)^2$. If  $D_1=0$ (that is, $c=4$; we operate under $c\neq 0$), then~\eqref{eq:inverse2}  has only the solutions $x=-2$ and the prior $x=0$.

\noindent
{\em Case $2$.} Let $b=c\neq 0,1$. Now, $x=-1$ is a solution of~\eqref{eq:inverse2}, so we next assume that $x\neq -1$  (also, $x\neq 0$, since this is not a solution unless $c=1$, which is impossible) . Multiplying~\eqref{eq:inverse2}  by $x(x+1)$ we get
$x-c(x+1)=cx(x+1)$, which is equivalent to $x^2+(2-c^{-1}) x+1=0$, which, by Lemma~\ref{lem10}$(ii)$ has two (respectively, one)  solutions if and only if $D_2=(2-c^{-1})^2-4=(1-4c) c^{-2}\in \left(\F_{p^n}\right)^2$ and $D_2\neq 0$ (respectively, $D_2=0$). Thus, altogether, we have three solutions for~\eqref{eq:inverse2} under  $0\neq 1-4c\in \left(\F_{p^n}\right)^2$ and two solutions if $c=4^{-1}$.

\noindent
{\em Case $3$.} Let $c\neq b\neq 0,1$ (so, $x\neq 0,-1$). Multiplying~\eqref{eq:inverse2}  by $x(x+1)$ we get $x-c(x+1)=bx(x+1)$, that is, $x^2+\left(\frac{b+c-1}{b}\right) x+\frac{c}{b}=0$. If $D_3=\left( \frac{b+c-1}{b}\right)^2-4\frac{c}{b}=0$, that is, $(b+c-1)^2=4bc$, then we have a unique solution, otherwise, we have two solutions if and only if
$0\neq D_3\in \left(\F_{p^n}\right)^2$, that is, $ 0\neq (b+c-1)^2-4bc\in \left(\F_{p^n}\right)^2$.

Below, we argue that we always can find some $b\neq 0,1,c$ for which $(b+c-1)^2-4bc\in \left(\F^*_{p^n}\right)^2$, except for $c=-1$, $p=3$, $n=2$, where, we can only find some values of $b$ for which $(b+c-1)^2-4bc=0$.

 If $c\neq -2,2,4$, then we can take $b=2^{-1} (c-2)$ and consequently, $b\neq 0,1,c$, and
 \[
 (b+c-1)^2-4bc=2^{-2} (c-4)^2\neq 0.
 \]


If $c=2$, or $c=4$, and $p\neq3, 5$, then, we can take $b=2(c+1)$, and $b\neq 0,1,c$, and
 \[
 (b+c-1)^2-4bc=(1-c)^2\neq0.
 \]
 Let $c=2$, and $p=3$. Then, $(b+c-1)^2-4bc=b^2+1$.
 If $n>2$, then we can take $b=\alpha-\alpha^{-1}$, where $\alpha$ is a primitive root of $\F_{3^n}$  (we here avoid the primitive polynomials $x^2\pm x-1=0$ over $\F_3$, since then, $b\in\{-1,1\}$). Consequently, $b\neq 0,1,-1$, and
 \[
 b^2+1=(\alpha+\alpha^{-1})^2\neq 0.
 \]
 Equation~\eqref{eq:inverse2} has then, in this case, two or fewer solutions.
 If $n=2$, then, we write $\F_{3^2}=\frac{\F_3[x]}{\langle x^2- x-1\rangle}=\F_3(\alpha)$, where $\alpha$ is a root of the primitive polynomial $x^2-x-1=0$.  Note that (recall Case 1) $D_1=2^2-4\cdot2=-4=2=(\alpha+1)^2\in\left(\F_{3^n}\right)^2\neq0$. Since, if $b\neq0,1,c$, equation~\eqref{eq:inverse2}  cannot have more than two solutions, we conclude that the $c$-differential uniformity is~$3$.


Let $c=2$,  and $p=5$. Then, $C_1=c^2-2c=-4=1=1^2\in\left(\F_{5^n}\right)^2\neq0$. Since, if $b\neq0,1,c$, equation~\eqref{eq:inverse2}  cannot have more than two solutions, the $c$-differential uniformity is 3 in this case.

 If $c=4$, and $p=3$, then $c=1$, which is a contradiction with  $c\neq1$.

  Let $c=4$ and $p=5$.  Then, $(b+c-1)^2-4bc=b^2+4$.
 If $n>2$, then, as before, we can take $b=\alpha-\alpha^{-1}$, where $\alpha$ is a primitive root of $\F_{5^n}$,  and consequently, $b\neq 0,1,-1$, and
  \[
 (b-2)^2+4b=b^2+4=(\alpha+\alpha^{-1})^2\neq 0.
 \]
  If $n=2$, then, we write $\F_{5^2}=\frac{\F_5[x]}{\langle x^2-x+2\rangle}=\F_5(\alpha)$, where $\alpha$ is a root of the primitive polynomial $x^2-x+2=0$. Taking $b=\alpha+3\neq 0,1,-1$, then (recall that $\alpha^2-\alpha+2-0$)
 \[
  (b-2)^2+4b=(2\alpha+2)^2\neq 0.
  \]

Let now $c=-2$. Note that, then $D_2=\left(\frac{3}{2}\right)^2\in\left(\F_{p^n}\right)^2$. Furthermore, for $p\neq3$, $D_2\neq0$, rendering at least three solutions of~\eqref{eq:inverse2} for $p\neq3$. If  $p=3$, $c=-2=1$, which is excluded from this theorem, since it is covered by the classical results.

Note that, by Cases 1, 2 and 3, equation~\eqref{eq:inverse2} cannot have more than three solutions, so the $c$-differential uniformity is always at most 3.

 The proof of the theorem is done.
  \end{proof}

\end{document}